\documentclass[runningheads]{llncs}

\usepackage{amsfonts}
\usepackage[T1]{fontenc}
\usepackage{graphicx}

\usepackage{etoolbox}
\usepackage{framed}
\usepackage{longtable}
\usepackage{multicol}
\usepackage{multirow}
\usepackage{bm}
\usepackage{proof}
\usepackage{caption}
\usepackage{wrapfig}
\usepackage{color}
\usepackage[dvipsnames,table]{xcolor}
\usepackage{enumerate}
\usepackage{enumitem}
\usepackage{fancyvrb}
\usepackage{float}
\usepackage{graphicx}
\usepackage{hyperref}
\usepackage[utf8]{inputenc}
\usepackage{listings}
\usepackage{mathtools}
\usepackage{relsize}
\usepackage{setspace}
\usepackage{stmaryrd}
\usepackage{subcaption}
\usepackage{cleveref}
\usepackage{tikz}
\usetikzlibrary{positioning}
\usetikzlibrary{matrix}
\usetikzlibrary{calc}
\usetikzlibrary{backgrounds}
\usetikzlibrary{patterns}
\usetikzlibrary{automata}

\usetikzlibrary{trees,arrows,arrows.meta,shapes,decorations.pathmorphing,backgrounds,positioning,fit,petri}
\usepackage[colorinlistoftodos,prependcaption,textsize=tiny]{todonotes}
\usepackage{url}
\usepackage[all]{xy}
\usepackage{proof}
\usepackage{syntax}
\usepackage{booktabs}
\usepackage{longtable}
\graphicspath{{images/}}
\tikzset{>={Latex[width=2mm,length=2mm]},
module/.style = {rectangle, draw, minimum height=0.8cm,
		minimum width=2.5cm, fill=orange!15, text centered,
		font=\ttfamily},
store/.style = {circle, draw, minimum height=1cm,
		fill=orange!15, text centered, font=\ttfamily},
triplearrow/.style={
		draw=black!75,
		color=black!75,
		double distance=3pt,
		postaction={draw=black!75, color=black!75},
		->},
}

\tikzstyle{pzgState} = [
rectangle split,
rectangle split horizontal,
rectangle split parts=2,
draw,
rectangle split part align={center,left},
node distance=3.5 cm,
rounded corners,
minimum height=0.5 cm,
minimum width=1.5 cm,
thick
]

\tikzstyle{pzgStateVar} = [
rectangle split,
rectangle split horizontal,
rectangle split parts=3,
draw,
rectangle split part align={center,left},
node distance=3.5 cm,
rounded corners,
minimum height=0.5 cm,
minimum width=1.5 cm,
thick
]

\tikzstyle{place}=[circle,thick,draw=blue!75,fill=blue!20,minimum size=6mm]
\tikzstyle{contour place}=[place,draw=red!100]
\tikzstyle{red place}=[place,draw=red!75,fill=red!20]
\tikzstyle{gray place}=[place,draw=black!100,fill=black!30]
\tikzstyle{transition}=[ rectangle,thick, fill=black, minimum width=8mm, inner ysep=2pt]
\tikzstyle{red transition}=[ rectangle,thick, fill=red!75, minimum width=8mm, inner ysep=2pt]
\tikzstyle{gray transition}=[ rectangle,thick, draw=black!100,fill=black!30, minimum width=8mm, inner ysep=2pt]
\tikzstyle{blue transition}=[ rectangle,thick, fill=blue!75, minimum width=8mm, inner ysep=2pt]
\tikzstyle{every label}=[black]

\tikzstyle{every node}=[initial text=]
\tikzstyle{location}=[rectangle, rounded corners, minimum size=12pt, draw=black, fill=blue!10, inner sep=2pt]
\tikzstyle{invariant}=[draw=black, dotted, inner sep=1pt]

\tikzstyle{final}=[double]

\newcommand\eg{e.g.,}

\newcommand{\red}[1]{\textcolor{red}{#1}}

\newcommand{\code}[1]{\texttt{#1}}

\lstdefinelanguage{Maude}{keywords={
    , mod, fmod, endm, endfm
    , view, endv
    , is
    , pr , protecting 
    , ex , extending 
    , inc, including
    , sort, sorts, subsort, subsorts
    , var, vars
    , op, ops
    , eq, ceq
    , rl, crl
    , if
    , search
    , check
    , such, that 
    , smt-search
    , rew
    , srew
    , dsrew
    , red, reduce
    , fth, endfth
    , then, else, fi
    },
    alsoletter={-}
}
\definecolor{darkred}{rgb}{0.6,0,0}
\definecolor{coloract}{rgb}{0.50,0.70,0.30}
\definecolor{colorclock}{rgb}{0.4,0.4,1}
\definecolor{colorconst}{rgb}{0.50, 0.20, 0.00}
\definecolor{colordisc}{rgb}{1, 0, 1}
\definecolor{colorloc}{rgb}{0.4,0.4,0.65}
\definecolor{colorparam}{rgb}{1,0.6,0.0}

\lstnewenvironment{maude}
{\lstset{language=Maude 
  , mathescape=true
  ,  keywordstyle=\color{blue}
  , basicstyle=\normalfont\ttfamily\singlespacing\footnotesize
  , comment=[l]{---},
  , commentstyle=\color{darkred}\ttfamily,
  , columns=flexible
  , numbers=none
  , showstringspaces=false
  , keepspaces=true
  , aboveskip=1pt
  , belowskip=1pt
  , lineskip = -0.25em
  }
}
{}

\long\def\comment#1{}

\newcommand{\To}{\Rightarrow}

\makeatletter

\newdimen\w@dth

\def\setw@dth#1#2{\setbox\z@\hbox{\scriptsize $#1$}\w@dth=\wd\z@
\setbox\@ne\hbox{\scriptsize $#2$}\ifnum\w@dth<\wd\@ne \w@dth=\wd\@ne \fi
\advance\w@dth by 1.2em}

\def\t@^#1_#2{\allowbreak\def\n@one{#1}\def\n@two{#2}\mathrel
{\setw@dth{#1}{#2}
\mathop{\hbox to \w@dth{\rightarrowfill}}\limits
\ifx\n@one\empty\else ^{\box\z@}\fi
\ifx\n@two\empty\else _{\box\@ne}\fi}}
\def\t@@^#1{\@ifnextchar_ {\t@^{#1}}{\t@^{#1}_{}}}

\def\t@left^#1_#2{\def\n@one{#1}\def\n@two{#2}\mathrel{\setw@dth{#1}{#2}
\mathop{\hbox to \w@dth{\leftarrowfill}}\limits
\ifx\n@one\empty\else ^{\box\z@}\fi
\ifx\n@two\empty\else _{\box\@ne}\fi}}
\def\t@@left^#1{\@ifnextchar_ {\t@left^{#1}}{\t@left^{#1}_{}}}

\def\two@^#1_#2{\def\n@one{#1}\def\n@two{#2}\mathrel{\setw@dth{#1}{#2}
\mathop{\vcenter{\hbox to \w@dth{\rightarrowfill}\kern-1.7ex
                 \hbox to \w@dth{\rightarrowfill}}}\limits
\ifx\n@one\empty\else ^{\box\z@}\fi
\ifx\n@two\empty\else _{\box\@ne}\fi}}
\def\tw@@^#1{\@ifnextchar_ {\two@^{#1}}{\two@^{#1}_{}}}

\def\tofr@^#1_#2{\def\n@one{#1}\def\n@two{#2}\mathrel{\setw@dth{#1}{#2}
\mathop{\vcenter{\hbox to \w@dth{\rightarrowfill}\kern-1.7ex
                 \hbox to \w@dth{\leftarrowfill}}}\limits
\ifx\n@one\empty\else ^{\box\z@}\fi
\ifx\n@two\empty\else _{\box\@ne}\fi}}
\def\t@fr@^#1{\@ifnextchar_ {\tofr@^{#1}}{\tofr@^{#1}_{}}}

\newdimen\W@dth
\def\setW@dth#1#2{\setbox\z@\hbox{$#1$}\W@dth=\wd\z@
\setbox\@ne\hbox{$#2$}\ifnum\W@dth<\wd\@ne \W@dth=\wd\@ne \fi
\advance\W@dth by 1.2em}

\def\T@^#1_#2{\allowbreak\def\N@one{#1}\def\N@two{#2}\mathrel
{\setW@dth{#1}{#2}
\mathop{\hbox to \W@dth{\rightarrowfill}}\limits
\ifx\N@one\empty\else ^{\box\z@}\fi
\ifx\N@two\empty\else _{\box\@ne}\fi}}
\def\T@@^#1{\@ifnextchar_ {\T@^{#1}}{\T@^{#1}_{}}}

\def\T@left^#1_#2{\def\N@one{#1}\def\N@two{#2}\mathrel{\setW@dth{#1}{#2}
\mathop{\hbox to \W@dth{\leftarrowfill}}\limits
\ifx\N@one\empty\else ^{\box\z@}\fi
\ifx\N@two\empty\else _{\box\@ne}\fi}}
\def\T@@left^#1{\@ifnextchar_ {\T@left^{#1}}{\T@left^{#1}_{}}}

\def\Tofr@^#1_#2{\def\N@one{#1}\def\N@two{#2}\mathrel{\setW@dth{#1}{#2}
\mathop{\vcenter{\hbox to \W@dth{\rightarrowfill}\kern-1.7ex
                 \hbox to \W@dth{\leftarrowfill}}}\limits
\ifx\N@one\empty\else ^{\box\z@}\fi
\ifx\N@two\empty\else _{\box\@ne}\fi}}
\def\T@fr@^#1{\@ifnextchar_ {\Tofr@^{#1}}{\Tofr@^{#1}_{}}}

\def\Two@^#1_#2{\def\N@one{#1}\def\N@two{#2}\mathrel{\setW@dth{#1}{#2}
\mathop{\vcenter{\hbox to \W@dth{\rightarrowfill}\kern-1.7ex
                 \hbox to \W@dth{\rightarrowfill}}}\limits
\ifx\N@one\empty\else ^{\box\z@}\fi
\ifx\N@two\empty\else _{\box\@ne}\fi}}
\def\Tw@@^#1{\@ifnextchar_ {\Two@^{#1}}{\Two@^{#1}_{}}}

\def\to{\@ifnextchar^ {\t@@}{\t@@^{}}}
\def\from{\@ifnextchar^ {\t@@left}{\t@@left^{}}}
\def\two{\@ifnextchar^ {\tw@@}{\tw@@^{}}}
\def\tofro{\@ifnextchar^ {\t@fr@}{\t@fr@^{}}}
\def\To{\@ifnextchar^ {\T@@}{\T@@^{}}}
\def\From{\@ifnextchar^ {\T@@left}{\T@@left^{}}}
\def\Two{\@ifnextchar^ {\Tw@@}{\Tw@@^{}}}
\def\Tofro{\@ifnextchar^ {\T@fr@}{\T@fr@^{}}}

\makeatother

\newcommand\tensor\otimes

\newcounter{marginalnote}
\newcommand{\highlight}[1]{\red{\texttt{#1}}}

\usepackage{graphicx}\usepackage{multirow}\usepackage{amsmath,amssymb,amsfonts}\usepackage{mathrsfs}\usepackage[title]{appendix}\usepackage{xcolor}\usepackage{textcomp}\usepackage{manyfoot}\usepackage{booktabs}\usepackage{algorithm}\usepackage{algorithmicx}\usepackage{algpseudocode}\usepackage{listings}\usepackage{url}

 \newcommand{\nats}{\mathbb{N}}

\newcommand{\oclos}[2]{\stackrel{#1}{#2}}
\newcommand{\SASYNC}{\square}
\newcommand{\SPAR}{{||}}
\newcommand{\SSYNC}{{s}}

\newcommand{\cU}{\mathcal{U}}
\newcommand{{\REL}}{\rightarrow}
\newcommand{{\AREL}}{\oclos{\SASYNC}{\REL}}
\newcommand{{\PREL}}{\oclos{\SPAR}{\REL}}
\newcommand{{\SREL}}{\oclos{\SSYNC}{\REL}}
\newcommand{\metaorel}[3]{{#1} \mathop{#2} {#3}}
\newcommand{\opair}[2]{\left\langle {#1} ; {#2} \right\rangle} 
\newcommand{\opset}[1]{\mathcal{P}\left({#1}\right)} 
\newcommand{\ocard}[1]{\mathrm{card}\left({#1}\right)} 
\newcommand{\orel}[2]{{#1} \mathop{\REL} {#2}}
\newcommand{\onrel}[2]{\metaorel{#1}{\not\REL}{#2}}
\newcommand{\ofunc}[3]{{#1} : {#2} \rightarrow {#3}}
\newcommand{\oobj}[2]{\left \langle {#1} : {#2} \right \rangle}

\graphicspath{ {./images/} }

\begin{document}
\mainmatter
\title{Unified Opinion Dynamic Modeling as Concurrent Set Relations in Rewriting Logic
\thanks{The work of Olarte, Ram\'irez, Rocha and Valencia was partially supported by the Minciencias (Ministerio de Ciencia
Tecnolog\'a e Innovaci\'on, Colombia) project PROMUEVA (BPIN 2021000100160).}}
\titlerunning{Opinion Dynamic Modeling in Rewriting Logic}
\author{Carlos Olarte\inst{1} \and Carlos Ramírez\inst{2} \and Camilo Rocha\inst{2} \and Frank Valencia\inst{3,2}}
\authorrunning{Olarte, Ram\'irez, Rocha and Valencia.}

\institute{
  LIPN, CNRS UMR 7030, Université Sorbonne Paris Nord, Villetaneuse, France \\
  \and
  Department of Electronics and Computer Science, Pontificia Universidad Javeriana, Cali, Colombia \\
  \and
  CNRS-LIX, École Polytechnique de Paris, France
}

\maketitle

\begin{abstract}
  Social media platforms have played a key role in weaponizing the
  polarization of social, political, and democratic processes. This
  is, mainly, because they are a medium for opinion formation.
Opinion dynamic models are a tool for understanding the role of
  specific social factors on the acceptance/rejection of opinions
  because they can be used to analyze certain assumptions on human
  behaviors.
This work presents a framework that uses concurrent set relations as
  the formal basis to specify, simulate, and analyze social
  interaction systems with dynamic opinion models. Standard models for
  social learning are obtained as particular instances of the proposed
  framework.
It has been implemented in the Maude system as a fully executable
rewrite theory that can be used to better
  understand how opinions of a system of agents can be shaped.
This paper also reports an initial exploration in Maude on the use
  of reachability analysis, probabilistic simulation, and statistical
  model checking of important properties related to opinion dynamic
  models.

\end{abstract}
 
\section{Introduction}
\label{sec.intro}

Social media platforms have played a key role in the polarization of
social, political, and democratic processes.
Social uprisings in the Middle East, Asia, and Central and South
America have led to sudden changes in the structure and nature of
society during this past
decade~\cite{lynch-arabspring-2015,asiafound-techconflict-2020,csis-milktea-2021,zolov-chile-2023,wikipedia-colombiauprising-2024,Kirby17}.
Polarization across the globe has paved the way to the divergence of
political attitudes away from the center, towards ideological
extremes, sometimes resulting in fractured institutions, erratic
policy making, incipient political dialog, and the resurgence of old
regimes~\cite{Garrett2017,iversen-masspolarization-2015,neverov-socialpubpol-2017,lynch-arabspring-2015,diplomat-southasia-2023,gupta-polmark-2023}.
Democracy, viewed as a system of power controlled by the people, has
been made vulnerable by severe polarization as opposing sides are seen
as adversaries that compete against an enemy needing to be
vanquished. As a result, popular election campaigns --including
presidential ones-- have compromised the basic principles of
democratic election in some
countries~\cite{beaufort-polarization-2018,suresh-tracking-2023,sarma-socialmediaelections-2023,ballard-polarizcongress-2023}.
All these scenarios have a common factor: social media interaction as
a medium for \emph{opinion formation} fueling polarization.

Social learning and opinion dynamic models have been developed to
understand the role of specific social factors on the
acceptance/rejection of opinions, such as the ones communicated via
social media (see,
e.g.,~\cite{Golup2017,alvim:hal-03872692,xia-opdynmod-2023,das-odmsocmedia-2014}). They
are often used to validate how certain assumptions on human behaviors
can explain alternative scenarios, such as opinion consensus, polarization 
and fragmentation.
In their micro-level approach, the one followed in the present work,
users are considered as agents that can share opinions on a given
topic. They update their opinion by interacting with a selected group
of users that have some influence on them (e.g., influencers, their family and
friends).
These dynamics take place at discrete time steps at which (some)
agents update their opinion. For instance, an opinion model can
deterministically update the opinion of all agents in such a
time-step, while another one can non-deterministically update the
opinion of a single agent.
Depending on the model of choice, which usually defines its own update
function for the individual agents, phenomena under different
assumptions can be observed. The ultimate goal is to understand how
the opinions of the agents, as a social system, are shaped after a
certain number of steps.

This work proposes a framework that uses concurrent set relations as
the formal basis to specify, simulate, and analyze social interaction
systems with dynamic opinion models.
The framework uses \textit{influence graphs} to specify the structure of agent
interactions in the social system under study: vertices represent
agents and a directed weighted edge from $a$ to $b$ represents the
weighted influence of agent's $a$ opinion over the opinion of agent
$b$.
In the sense of set relations in~\cite{DBLP:journals/tcs/RochaMD11},
the framework comprises two main mechanisms that are combined via
closures for specifying opinion dynamics over the graphs: namely, an
atomic set relation and a strategy. The \textit{atomic set relation}
updates the opinion of a single vertex w.r.t. a set of edges (and the
corresponding vertices) incident to it.  The \textit{strategy} selects
the edges that will be used to update in parallel (i.e.,
synchronously) the opinion associated to the vertices with edges
incident to it in the given set.
As a consequence, dynamic opinion models can be formalized as a
concurrent set relation system, with parametric update function, using
the composition of an atomic relation and a strategy via closures.
An important observation is that the determinism or non-determinism
inherent to a given opinion dynamic model is exactly captured by the
deterministic or non-deterministic nature of the corresponding
concurrent set relation.

Standard models for social learning are obtained as particular
instances of the proposed framework. The classical DeGroot opinion
model~\cite{degroot-opinionmodel-1974} is obtained as the synchronous
closure under the maximal redices strategy of a given atomic set
relation. In a similar fashion, gossip-based models that use pairwise
interactions to represent the opinion formation process (see,
e.g.,~\cite{fagnani-gossipopmodel-2007}) are obtained via the
asynchronous closure where the strategy selects single edges for the
given atomic set relation. Other opinion models can be obtained via
the synchronous closure of an atomic set relation, as midpoints
between De Groot and gossip-based models.

The proposed framework has been implemented in the Maude
system~\cite{clavel-maudebook-2007}. It is a rewriting logic theory
that exploits the reflective capabilities of rewriting logic and that
can be particularized to the opinion model of interest.
A state is an object-like configuration representing the structure of
the system and its opinion values. An object is either an agent $u$
with its opinion $o_u$, specified as $\oobj{u}{o_u}$, or the influence
of agent $u$ over agent $v$ with weight $i_{uv}$, specified as
$\oobj{(u, v)}{i_{uv}}$. The update function $\mu$ of each specific model is
to be defined equationally.
The implementation of both the atomic set relation and the strategy is
inspired by the ideas in~\cite{DBLP:journals/scp/RochaM14}.
The atomic set relation is axiomatized as a (non-executable) rewrite
rule that takes as input an agent $\oobj{u}{o_u}$ and a set of edges
$A \subseteq E$ in the current state. For a given state, it updates
the opinion $o_u$ to a new opinion $o_u'$ using $\mu$, and the opinion
and influence of agents adjacent to it w.r.t. $A$. As a result, each
atomic step rewrites a single object $\oobj{u}{o_u}$ to its updated
version $\oobj{u}{o_u'}$.
The metalevel is used to apply the atomic rewrite rule over the agents
in a state according to the edges selected by the given strategy: only
agents appearing as targets of the directed edges have their opinion
updated. This strategy is defined equationally by the user and
computes a collection of subsets of $E$: a parallel rewrite step under
the maximal redices strategy is performed for each subset $A$ of
edges. Since the atomic rewrite relation is deterministic, the
strategy is the only source of non-determinism in the system and a
concurrent step is made for each identified subset $A$.

The implementation of the proposed framework results in a fully
executable top-most object-like rewrite theory in Maude that can be
used to better understand how opinions of a system of agents are
shaped --and to ultimately understand polarization--- using formal
methods techniques, such as reachability analysis and temporal model
checking.

This work is part of a broader effort to make available computational
ideas and approaches for analyzing phenomena in social networks, such
as polarization, consensus, and fragmentation. They include
concurrency models, modal and probabilistic logics, and formal methods
frameworks, techniques, and tools. In this context, the work presented
here is a first step towards the use of rewriting logic for such
purposes. As it is explained in the sections that follow, one major
problem a opinion dynamic model may face is that of state
explosion. An initial exploration on the use of probabilistic
simulation and statistical analysis is reported in this work. However,
the extension of the proposed framework to a fully probabilistic
setting, in which --e.g.-- the strategy selects the set of edges
according to a probability distribution function, falls outside the
scope of this work.  It needs to be further explored as future work as
it may open the door to statistical model checking of novel properties
using a new breed of measures and thus pave the way to the analysis of
quantitative properties beyond the reach of techniques currently
available for opinion dynamic models.

\paragraph{Organization.} After recalling the notion of set relations in
\Cref{sec.rels}, \Cref{sec.models} shows how different models for
social learning can be seen as particular instances (atomic set
relation and strategy) of this framework. The implementation in Maude
is described in \Cref{sec.rl}, while different analyses performed on
the proposed rewrite theory are introduced in
\Cref{sec.exp}. \Cref{sec.concl} concludes the paper.
The full Maude specification supporting the set relations framework
is available at~\cite{tool}, as companion
tool to the paper.


\section{Overview of Rewriting Logic and Maude}
\label{sec.prem}

A \emph{rewrite theory} \cite{meseguer-rltcs-1992} is a tuple $\mathcal{R} =
(\Sigma, E, L, R)$ such that: $(\Sigma, E)$ is an equational theory where
$\Sigma$ is a signature that declares sorts, subsorts, and function symbols;
$E$ is a set of (conditional) equations of the form $t=t' \mbox{ \textbf{if} }
\psi$, where $t$ and $t'$ are terms of the same sort, and $\psi$ is a
conjunction of equations; $L$ is a set of \emph{labels}; and $R$ is a set of
labeled (conditional) rewrite rules of the form $l : q \longrightarrow r \mbox{
\textbf{if} } \psi$, where $l \in L$ is a label, $q$ and $r$ are terms of the
same sort, and $\psi$ is a conjunction of equations. Condition $\psi$ in
equations and rewrite rules can be more general than conjunction of equations,
but this extra expressiveness is not needed in this paper. 
    
$T_{\Sigma, s}$ denotes the set of ground terms of sort $s$, and
$T_{\Sigma}(X)_s$ denotes the set of terms of sort $s$ over a set of  sorted
variables $X$. $T_{\Sigma}(X)$ and $T_{\Sigma}$ denote all terms and ground
terms, respectively. A substitution $\sigma : X \rightarrow T_{\Sigma}(X)$ maps
each variable to a term of the same sort, and $t \sigma$ denotes the term
obtained by simultaneously replacing each variable $x$ in a term $t$ with
$\sigma(x)$. 

A \emph{one-step rewrite} $t \longrightarrow_{\mathcal{R}} t'$ holds if there
is a rule $l : q \longrightarrow r \mbox{ \textbf{if} } \psi$, a subterm $u$ of
$t$, and a substitution $\sigma$ such that $u = q\sigma$ (modulo equations),
$t'$ is the term obtained from $t$ by replacing $u$ with $r\sigma$, and
$v\sigma = v'\sigma$ holds for each $v = v'$ in $\psi$. The
reflexive-transitive closure of $\longrightarrow_{\mathcal{R}}$ is denoted as
$\longrightarrow_{\mathcal{R}}^\ast$.

A rewrite theory $\mathcal{R}$ is called \emph{topmost} iff there is a sort
$\mathit{State}$ at the top of one of the connected components of the subsort
partial order such that for each rule $l : q \longrightarrow r \mbox{
\textbf{if} } \psi$, both $q$ and $r$ have the top sort $\mathit{State}$, and
no operator has sort $\mathit{State}$ or any of its subsorts as an argument
sort.

Maude~\cite{clavel-maudebook-2007} is a language and tool
supporting the specification and analysis of  rewrite theories. A Maude module
(\texttt{\textcolor{blue}{mod}} \textit{M} \texttt{\textcolor{blue}{is}} ...
\texttt{\textcolor{blue}{endm}}) specifies a rewrite theory $\mathcal{R}$.
Sorts and subsort relations are declared by the keywords
\texttt{\textcolor{blue}{sort}}  and \texttt{\textcolor{blue}{subsort}}; 
function symbols, or \emph{operators}, are introduced with the
\texttt{\textcolor{blue}{op}} keyword: \texttt{\textcolor{blue}{op}} $f$
\texttt{:} $s_1$ ... $s_n$ \texttt{->} $s$, where $s_1$, \ldots,  $s_n$ are the
sorts of its arguments, and $s$ is its (value) sort. Operators can have
user-definable syntax, with underbars `\verb@_@' marking each of the argument
positions (\eg{} \verb@_+_@). Some operators can have equational
attributes, such as \texttt{\textcolor{blue}{assoc}},
\texttt{\textcolor{blue}{comm}}, and \texttt{\textcolor{blue}{id:}}$\;t$,
stating that the operator is, respectively, associative,  commutative,  and/or
has identity element $t$. 
Equations are specified with the syntax 
 \texttt{\textcolor{blue}{eq}} $t$ \texttt{=} $t'$  or
    \texttt{\textcolor{blue}{ceq}} $t$ \texttt{=} $t'$
      \texttt{\textcolor{blue}{if}} $\psi$;
      and 
  rewrite rules as \texttt{\textcolor{blue}{rl}}
    \texttt{[}$l$\texttt{]\,:} $u$ \texttt{\,=>\,} $v$ or
    \texttt{\textcolor{blue}{crl}} \texttt{[}$l$\texttt{]\,:} $u$
      \texttt{\,=>\,} $t'$ \texttt{\textcolor{blue}{if}}
        $\psi$.
The mathematical variables
  in such statements are declared with the keywords
  \texttt{\textcolor{blue}{var}} and \texttt{\textcolor{blue}{vars}}.

 Maude provides a large set of analysis methods,  including
computing the normal form of a term $t$ 
(command \lstinline[mathescape]{red $t$}), simulation by rewriting 
(\lstinline[mathescape]{rew $t$}),
reachability analysis (\lstinline[mathescape]{search $t$ =>* $t'$ such that $\psi$}), 
and rewriting
according to a given rewrite strategy (\lstinline[mathescape]{srew $t$ using $str$}). 
Basic such rewrite strategies include $r\mathtt{[}\sigma\mathtt{]}$
(apply  rule with label $r$ once with the optional ground substitution $\sigma$),
\code{idle} (identity), \code{fail}
(empty set), and \code{match $P$ s.t.\ $C$}, which  checks whether the current
term matches the pattern $P$ subject to the constraint $C$. Compound strategies
can be defined using concatenation ($\alpha\,;\,\beta$), disjunction ($\alpha\,
|\, \beta$), iteration ($\alpha \,\mathtt{*}$), $\alpha \code{ or-else } \beta$
(execute $\beta$ if $\alpha$ fails), etc. 

The Unified Maude model-checking tool \cite{DBLP:journals/jlap/RubioMPV21}
(\texttt{umaudemc}) allows for the use of different model checkers to analyze
Maude specifications. Besides being an interface for the standard LTL model
checker of Maude, it also offers the possibility of interfacing external CTL
and probabilistic model checkers. For the purpose of this paper, the command
\code{scheck} \cite{DBLP:conf/fm/RubioMPV23} is used to assign probabilities to
the transition system generated by an initial term $t$, and perform statistical
model checking to estimate quantitive expressions written in the QuaTEx
language. Hence, it is possible to compute, e.g., the expected value of the
number of communication or interactions needed to reach a consensus in a
network.

\paragraph{Meta-programming.} Maude supports \emph{meta-programming}, where a
Maude module $M$ (resp., a term $t$) can be (meta-)represented as a Maude
\emph{term} $\overline{M}$ of sort \code{Module} (resp.\  as a Maude term
$\overline{t}$ of sort \code{Term}) in Maude's \code{META-LEVEL} module. Maude
provides built-in functions such as  \code{metaRewrite}, and
\code{metaSearch}, which are the ``meta-level'' functions corresponding to
``user-level'' commands to perform 
rewriting and search, respectively.

\section{Set Relations}
\label{sec.rels}

This section introduces set relations and their notation, as used in
this paper. It defines the asynchronous, parallel, and synchronous set
relations as closures of an atomic set relation. This section is
based, mainly, on~\cite{DBLP:journals/tcs/RochaMD11}.

Let $\cU$ be a set whose elements are denoted $A, B, \ldots$ and let
$\REL$ be a binary relation on $\cU$.
An element $A$ of $\cU$ is called a $\REL$-\textit{redex} iff there
exists $B \in \cU$ such that the \textit{pair} $\opair{A}{B}
\mathop{\in} \REL$. The expressions $\orel{A}{B}$ and $\onrel{A}{B}$
denote $\opair{A}{B} \mathop{\in} \REL$ and $\opair{A}{B}
\mathop{\not\in} \REL$, respectively.
The \textit{identity} and \textit{reflexive-transitive} closures of
$\REL$ are defined as usual and denoted $\stackrel{0}{\REL}$ and
$\stackrel{*}{\REL}$, respectively.

It is assumed that $\cU$ is the family of all \textit{nonempty} finite
subsets of an abstract and possibly infinite set $T$ whose members
are called $\textit{elements}$ (i.e., $\cU \subseteq \opset{T}$,
$\emptyset\not\in\cU$, and if $A \in \cU$, then $\ocard{A} \in \nats$
). Therefore, $\REL$ is a binary relation on finite subsets of
elements in $T$.
When it is clear from the context, curly brackets are omitted from set
notation; e.g., $\orel{a,b}{b}$ denotes $\orel{\{a,b\}}{\{b}\}$.
Because this convention, the symbol `,' is overloaded to denote set
union. For example, if $A$ denotes the set $\{a,b\}$, $B$ the set
$\{c,d\}$, and $D$ the set $\{d,e\}$, the expression $\orel{A,B}{B,D}$
denotes the pair $\orel{a,b,c,d}{c,d,e}$.

Given a set of elements, in the asynchronous set relation exactly one
redex is selected to be updated.

\begin{definition}[Asynchronous Set Relation]\label{def.rels.async}
  The \emph{asynchronous relation} $\AREL$ is defined as the
  asynchronous closure of $\REL$, i.e., the set of pairs $\opair{A}{B}
  \in \cU \times \cU$ such that $\metaorel{A}{\AREL}{B}$ iff there
  exists a $\REL$-redex $A' \subseteq A$ and an element $B' \in \cU$
  such that $\orel{A'}{B'}$ and $B = (A \setminus A') \cup B'$.
\end{definition}

In the parallel set relation, a non-empty collection of redices is
identified to be updated in parallel (i.e., without interleaving).

\begin{definition}[Parallel Set Relation]\label{def.rels.parallel}
  The \emph{parallel relation} $\PREL$ is defined as the parallel
  closure of $\REL$, i.e., the set of pairs $\opair{A}{B} \in \cU
  \times \cU$ such that $\metaorel{A}{\PREL}{B}$ iff there exist
  $\REL$-redices $A_1,\ldots,A_n \subseteq A$ (nonempty) pairwise
  disjoint and elements $B_1,\ldots,B_n$ in $\cU$ such that
  $\orel{A_i}{B_i}$, for $1 \leq i \leq n$, and $B = \left(A \setminus
  \bigcup_{1 \leq i \leq n}A_i \right) \cup \left(\bigcup_{1 \leq i
    \leq n} B_i\right)$.
\end{definition}

The synchronous set relation $\SREL$ applies as many atomic reductions
as possible, in parallel. However, in contrast to the previous two
closures, the redices are selected with the help of a strategy $s$,
namely, a function that identifies a non-empty subset of redices.  As
a consequence, the synchronous set relation is a subset of the
parallel set relation.
It is important to note that the notion of strategy used for defining
the synchronous closure of the atomic set relation is different to the
one introduced in~\Cref{sec.intro} for the framework; the name used in
this section is kept from~\cite{DBLP:journals/tcs/RochaMD11}.

\begin{definition}[$\REL$-strategy]\label{def.rels.strat}
  A \emph{$\REL$-strategy} is a function $s$ that maps any element $A
  \in \cU$ into a set $s(A) \subseteq \opset{\REL}$ such that
  if $s(A) = \{\opair{A_1}{B_1}, \ldots \opair{A_n}{B_n}\}$,
  then $A_i \subseteq A$ and $\orel{A_i}{B_i}$, for $1 \leq i \leq n$,
  and $A_1, \ldots, A_n$ are pairwise disjoint.
\end{definition}

\begin{definition}[Synchronous Relation]\label{def.rels.sync}
  Let $s$ be a $\REL$-strategy. The \emph{synchronous relation}
  $\SREL$ is defined as the synchronous closure of $\REL$ w.r.t. $s$,
  i.e., the set of pairs $\opair{A}{B} \in \cU \times \cU$ such that
  $\metaorel{A}{\SREL}{B}$ iff $B = \left(A \setminus \bigcup_{1 \leq
    i \leq n}A_i \right) \cup \left(\bigcup_{1 \leq i \leq n}
  B_i\right)$ where $s(A) = \{\opair{A_1}{B_1}, \ldots
  \opair{A_n}{B_n}\}$.
\end{definition}

This section is concluded with an example that illustrates the notions
introduced so far.
\subsubsection{Vaccine Example.} 
Consider the directed weighted graph $G = (V, E, i)$ in
Figure~\ref{fig.rels.vaccines}.  It represents a social system with
six agents $V = \{a,b,c,d,e,f\}$ and twelve opinion influences. The
label $i(u,v)$ associated to each edge $(u, v)$ from agent $u$ to
agent $v$ denotes the opinion influence $i_{uv} = i(u,v)$ of agent $u$
over the opinion of agent $v$ (about a given topic): these values are
in the real interval $[0,1]$ (i.e., $\ofunc{i}{E}{[0,1]}$); the higher
the value, the stronger the influence. In this example, the influence
of $f$ over $a$ is the strongest possible. Notice that agents may also
have \emph{self-influence}, representing agents whose opinion need not
be completely influenced by the opinion of the others.

The initial opinions (or beliefs) of the agents are depicted within
the box below each node. They are specified by a function
$\ofunc{o}{V}{[0,1]}$, which is assumed to represent the opinion value
$o_u = o(u)$ of each agent $u$ on the given topic. The greater the
value, the stronger (weaker) the agreement (disagreement) with the
proposition, and $0$ represents total disagreement.  In this example
such a proposition is \emph{vaccines are safe}. Intuitively, the
agents $a$, $b$, and $c$ are in strong disagreement with vaccines
being safe (the anti-vaxxers) and the rest are in strong agreement
(the pro-vaxxers).

Notice that although $a$ is the most extreme anti-vaxxer, the most
extreme pro-vaxxer $f$ has a strong influence over $a$. Hence, it is
expected that the evolution of $a$'s opinion will be highly influenced
by the opinion of $f$. In general, an agent's opinion evolution takes
into account a subset of its influences, as will be explained shortly.

\begin{figure}[htbp]
  \centering
  \includegraphics[scale=0.9]{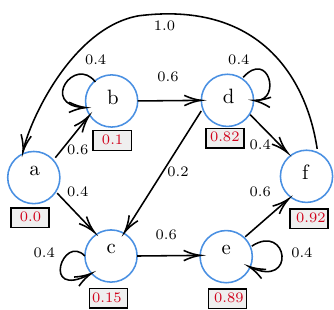}
  \caption{Graph representing opinion and influence interaction in a
    social system. Initial opinions are given within the box below
    each node. The labels on each edge $(u,v)$ represent the influence
    value of agent $u$ over agent $v$.}
  \label{fig.rels.vaccines}
\end{figure}

Recall the object-like notation in~\Cref{sec.intro}.  The set of
elements $T$ is made of pairs of the form $\oobj{u}{r}$ or $\oobj{(u,
  v)}{r}$, with $u,v \in V$, $(u,v) \in E$, and $r \in [0,1]$.
The graph in Figure~\ref{fig.rels.vaccines} can be
specified as the set of elements $\Gamma$:
{\small\begin{align*} \Gamma = \{ & \oobj{a}{0.0}, \oobj{b}{0.1},
  \oobj{c}{0.15}, \oobj{d}{0.82}, \oobj{e}{0.89}, \oobj{f}{0.92},\\ &
  \oobj{(a,b)}{0.6}, \oobj{(a,c)}{0.4}, \oobj{(b,d)}{0.6},
  \oobj{(c,e)}{0.6}, \oobj{(d,c)}{0.2}, \oobj{(d,f)}{0.4}, \\ &
  \oobj{(e,f)}{0.6}, \oobj{(f,a)}{1.0}, \oobj{(b,b)}{0.4},
  \oobj{(c,c)}{0.4}, \oobj{(d,d)}{0.4}, \oobj{(e,e)}{0.4} \}.
\end{align*}}

The atomic relation $\REL_A$ is defined over elements representing
agents and is parametric on a set $A$ of elements representing edges
in $\Gamma$.  In this example, it follows the pattern
\begin{equation}\label{eq.rels.atomic}
  \oobj{u}{o_u} \REL_A \oobj{u}{\sum_{\oobj{(x,u)}{i_{xu}} \in A} o_x \cdot \frac{i_{xu}}{\sum_{\oobj{(y,u)}{i_{yu}}\in A} i_{yu}}},
\end{equation}
where the summation in the denominator is assumed to be non-zero.
The opinion $o_u$ of an agent $u$ w.r.t. to $A$ is updated to be the
weighted average of the opinion values of those agents adjacent to $u$
and whose influence is present in $A$.
For instance, let $A = \{\oobj{(a,b)}{0.6}, \oobj{(b,b)}{0.4},
\oobj{(c,e)}{0.6}\}$. Then, the atomic set relation $\REL_A$ has the
following two pairs:
\begin{align*}
  \oobj{b}{0.1} \REL_A \oobj{b}{0.04} && \oobj{e}{0.89} \REL_A \oobj{e}{0.15}.
\end{align*}
In the case of agent $b$, its opinion is updated to $0.04 = 0.0 \cdot
\frac{0.6}{1.0} + 0.1\cdot\frac{0.4}{1.0}$ because, w.r.t. $A$, it is
influenced both by itself and by agent $a$, whose opinion value is
$0.0$ and influence over $b$ is $0.6$.
In the case of agent $e$, its opinion is influenced only by agent
$c$. The value is updated to $0.15 = 0.15 \cdot \frac{0.6}{0.6}$.  It
can be said that, w.r.t. $A$, agent $e$ acts like a \emph{puppet} whose own
opinion is not taken into account when it is updated.

The asynchronous closure of $\REL_A$ has exactly two pairs, one for
each redex determined by $\REL_A$ (i.e., one for agent $b$ and another
for agent $e$):
\begin{align*}
 & \Gamma \AREL_A \left (\Gamma \setminus \{ \oobj{b}{0.1} \} \right) \cup \{\oobj{b}{0.04}\} &\quad 
 & \Gamma \AREL_A \left (\Gamma \setminus \{ \oobj{e}{0.89} \} \right) \cup \{\oobj{e}{0.15}\}.
\end{align*}
The parallel closure $\PREL$ has three pairs: one in which the
opinions of both $b$ and $e$ are updated, in addition to the same two
pairs present in the asynchronous closure:
\[
    \begin{array}{lll}
  \Gamma \PREL_A \left (\Gamma \setminus \{ \oobj{b}{0.1} \} \right) \cup \{\oobj{b}{0.04}\} &\qquad\quad&
  \Gamma \PREL_A \left (\Gamma \setminus \{ \oobj{e}{0.89} \} \right) \cup \{\oobj{e}{0.15}\} \\
  \multicolumn{3}{l}{\Gamma \PREL_A \left (\Gamma \setminus \{ \oobj{b}{0.1}, \oobj{e}{0.89} \} \right) \cup \{\oobj{b}{0.04}, \oobj{e}{0.15}\}.}
    \end{array}
\]
Finally, to illustrate the synchronous closure of $\REL$,
let $s = A$ be the strategy. That is, all redices in $\REL_A$
are identified to be reduced. Therefore, this relation has the only
pair in which the opinions of both $b$ and $e$ are updated in
parallel:
\begin{align*}
 & \Gamma \SREL_A \left (\Gamma \setminus \{ \oobj{b}{0.1}, \oobj{e}{0.89} \} \right) \cup \{\oobj{b}{0.04}, \oobj{e}{0.15}\}.
\end{align*}
 
\section{Opinion Dynamic Models}
\label{sec.models}

This section shows how opinion dynamic models can be specified as set
relations (see ~\Cref{sec.rels}). In particular, a gossip-based and
the classical De Groot opinion models are introduced, as well as a
generalization of De Groot and gossip (under some conditions), here
called the \emph{hybrid opinion model}.

The three above-mentioned models are defined, as stated in
\Cref{sec.rels}, over a directed weighted graph $G = (V, E, i)$
representing a social system, with agents $V$, directed opinion
influences $E \subseteq V \times V$, and influence values
$\ofunc{i}{E}{[0,1]}$. A given topic (i.e., proposition) is fixed.
The weight $i_{uv}=i(u,v)$ associated to each edge $(u, v)\in E$ from
agent $u$ to agent $v$ denotes the opinion influence value of agent
$u$ over the opinion value of agent $v$ on the given topic.
The  opinion value $o_u = o(u) \in [0,1]$ associated to each agent $u
\in V$ in the given topic is assumed to be known by all agents in the
system.
As in \Cref{sec.rels}, the higher the value of a opinion (resp. influence),
the stronger the agreement (resp. influence).

The set of elements $T$ in the set relations framework
(see~\Cref{sec.rels}) is made of pairs of the form $\oobj{u}{r}$ or
$\oobj{(u, v)}{r}$, with $u,v \in V$, $(u,v) \in E$, and $r \in
[0,1]$.
A $G$-\textit{configuration} (or \textit{configuration}) is the set of
elements in $T$ that exactly represent the structure of $G$, and the
values of opinions and interactions. Therefore, in the rest of this
section, it is assumed that any configuration $\Gamma$ can be
partitioned in two sets $\Gamma_o$ and $\Gamma_i$, respectively
containing elements of the form $\oobj{u}{o_u}$ specifying opinions
and $\oobj{(u, v)}{i_{uv}}$ specifying influences.

A model specifies how opinions (associated to agents) can be
updated.
Each model definition comprises three pieces; namely, an atomic
relation, a strategy, and an update function for opinions.
Therefore, a model specifies how a $G$-configuration $\Gamma =
\Gamma_o \cup \Gamma_i$ can change to another $G$-configuration
$\Gamma' = \Gamma_{o'} \cup \Gamma_i$, where only opinions are
updated.
It is important to note that the notion of strategy introduced
in this section generalizes the notion of strategy introduced
in~\Cref{sec.rels}, as will be explained later.

The atomic relation is defined in \Cref{sec.models.atomic} for the
three models. Each model is introduced by identifying a specific
strategy and a specific update function in subsequent sections.

\subsection{The Atomic Relation}
\label{sec.models.atomic}

The atomic relation $\REL_A$ is parametric on a subset $A \subseteq
\Gamma_i$ and defines how the opinion of a single agent may evolve.
The set of influences $A$ directly identifies the influences (and
indirectly the opinions) to update the opinion of each agent in the
configuration $\Gamma$ (i.e., in $\Gamma_o$).
For each one of the three models, the
atomic relation $\REL_A$ follows the pattern:
\begin{equation}\label{eq.models.atomic}
  \oobj{u}{o_u} \REL_A \oobj{u}{\mu(\Gamma, A, u)},
\end{equation}
where $\mu$ is the update function specific to each model.  This
function takes as input a $G$-configuration (e.g., $\Gamma$), a subset
of its influences (e.g., $A$), and the agent whose opinion is to be
updated (e.g., $u$), and outputs the new opinion for agent $u$
w.r.t. $\Gamma$ and $A$ in the corresponding model.

\subsection{Gossip-based Models}
\label{sec.models.gossip}

In a gossip-based model, single peer-to-peer interactions are used to
update the opinion of a single user at each time-step.
In general, a strategy in the proposed framework identifies a
collection of subsets of interactions in $\Gamma_i$. In particular,
the strategy $\rho_\text{gossip}$ maps a $G$-configuration to the
collection of singletons made from the influences in $\Gamma_i$:
\begin{align*}
  & \rho_\text{gossip}(\Gamma)  = \{ \{ x \} \mid x \in \Gamma_i \}.
\end{align*}
This means that, at each time-step, the opinion value of agent $v$ can
be updated w.r.t. the opinion value of agent $u$ for each singleton
$\{\oobj{(u,v)}{i_{uv}}\}$ computed by the strategy
$\rho_\text{gossip}(\Gamma)$.

The update function $\mu_\text{gossip}$ is defined for any $u \in V$
and $A=\{\oobj{(v, u)}{i_{vu}}\} \in \rho_\text{gossip}(\Gamma)$ as:
\begin{align*}
  & \mu_\text{gossip}(\Gamma, A, u)  = o_u + (o_v-o_u)\cdot i_{vu}.
\end{align*}

Each singleton $A \in \rho_\text{gossip}(\Gamma)$ determines an atomic
relation that updates exactly one agent's opinion in the given
configuration.
Recall, from~\Cref{sec.models.atomic}, that each pair in the atomic
set relation $\REL_A$ has the form:
\begin{align*}
  \oobj{u}{o_u} \REL_A \oobj{u}{\mu_\text{gossip}(\Gamma, A, u)}.
\end{align*}
Hence, in this model, the opinion of an agent $u$ is updated by
identifying an edge from other agent $v$ with influence $i_{vu}$ over
$u$ and by adding to its current opinion $o_u$ the weighted difference
of opinion $(o_v-o_u)\cdot i_{vu}$ of $v$ over $u$.

A gossip-based model is identified as a binary set relation on
$G$-configurations in terms of the asynchronous closure of $\REL_A$,
for each singleton $A\in \rho_\text{gossip}(\Gamma)$.

\begin{definition}\label{def.models.gossip}
The $\REL_\text{gossip}$ set relation is the set of pairs
$\opair{\Gamma}{\Gamma'}$ of $G$-configurations such that:
\begin{align*}
  \Gamma \REL_\text{gossip} \Gamma' \quad\text{iff} \quad (\exists A
  \in \rho_\text{gossip}(\Gamma)) \; \Gamma \AREL_A \Gamma'.
\end{align*}
\end{definition}

From the viewpoint of concurrency, the gossip-based opinion dynamic
model captured by $\REL_\text{gossip}$ is non-deterministic in the
sense that at each state (i.e., $G$-configuration) exactly
$|\Gamma_i|$ transitions are possible, one per edge in $E$.

\subsection{De Groot}
\label{sec.models.degroot}

In the De Groot model, the opinion value of every agent in the network
is updated at each time-step. All influences are considered at the
same time.

The strategy for De Groot in the proposed framework identifies the
whole set of interactions in the network, i.e., $\Gamma_i$.
In particular, the strategy $\rho_\text{DeGroot}$ maps a
$G$-configuration to the singleton whose only element is $\Gamma_i$:
\begin{align*}
  & \rho_\text{DeGroot}(\Gamma)  = \{ \Gamma_i \}.
\end{align*}

The update function $\mu_\text{DeGroot}$ is defined for any $u \in V$
and $A \in \rho_\text{DeGroot}(\Gamma)$ (i.e., $A = \Gamma_i$) as:
\begin{align*}
  & \mu_\text{DeGroot}(\Gamma, A, u)  = o_u + \sum_{\oobj{(v,u)}{i_{vu}} \in A}(o_v-o_u)\cdot \frac{i_{vu}}{\sum_{\oobj{(x,u)}{i_{xu}} \in A}i_{xu}},
\end{align*}
where the summation in the denominator is assumed to be non-zero.
Otherwise, the value of this function is assumed to be $o_u$ (i.e.,
the opinion of agent $u$ does not change).

The De Groot model is identified as a binary set relation on
$G$-configurations in terms of the synchronous closure of
$\REL_{\Gamma_i}$ under the maximal redices strategy for $s =
\Gamma_i$.

\begin{definition}\label{def.models.degroot}
The $\REL_\text{DeGroot}$ set relation is the set of pairs
$\opair{\Gamma}{\Gamma'}$ of $G$-configurations such that:
\begin{align*}
  \Gamma \REL_\text{DeGroot} \Gamma' \quad\text{iff} \quad \Gamma \oclos{\Gamma_i}{\REL}_{\Gamma_i} \Gamma'.
\end{align*}
\end{definition}

From the viewpoint of concurrency, the De Groot opinion dynamic model
captured by $\REL_\text{DeGroot}$ is deterministic in the sense that,
at each state, there is exactly only one possible transition where all
influences are taken into account to update each agent's opinion
without interleaving.

\subsection{The Hybrid Model}
\label{sec.models.hybrid}

The hybrid model considers every possible influence scenario in the
network, i.e., any possible combination of influences are used to
update the opinion of agents that may be affected by them at each
time-step.
Therefore, the strategy in the proposed framework identifies all
non-empty subsets of interactions in $\Gamma_i$. In particular,
the strategy $\rho_\text{hybrid}$ maps a $G$-configuration to the
collection of non-empty subsets made from the influences in $\Gamma_i$:
\begin{align*}
  & \rho_\text{hybrid}(\Gamma)  = \{ A \mid A \subseteq \Gamma_i \mbox{ and } A\neq \emptyset \}.
\end{align*}
This means that, at each time-step, the opinion value of an agent $v$ can
be updated with a subset of its influencers.

The update function $\mu_\text{hybrid}$ is the same as function
$\rho_\text{DeGroot}$. That is, it is defined for any $u \in V$ and $A
\in \rho_\text{hybrid}(\Gamma)$ as:
\begin{align*}
  & \mu_\text{hybrid}(\Gamma, A, u)  = o_u + \sum_{\oobj{(v,u)}{i_{vu}} \in A}(o_v-o_u)\cdot \frac{i_{vu}}{\sum_{\oobj{(x,u)}{i_{xu}} \in A}i_{xu}},
\end{align*}
where the summation in the denominator is assumed to be non-zero.
Otherwise, the value of this function is assumed to be $o_u$ (i.e.,
the opinion of agent $u$ does not change).
Each subset $A \in \rho_\text{hybrid}(\Gamma)$ determines an atomic
relation that may update more that one agent's opinion.
Hence, in this model, the opinion of an agent is updated by
identifying some edges that may have influence over it. 

The hybrid model is identified as a binary set relation on
$G$-configurations in terms of the synchronous closure of $\REL_A$,
for each subset $A\in \rho_\text{hybrid}(\Gamma)$.

\begin{definition}\label{def.models.hybrid}
The $\REL_\text{hybrid}$ set relation is the set of pairs
$\opair{\Gamma}{\Gamma'}$ of $G$-configurations such that:
\begin{align*}
  \Gamma \REL_\text{hybrid} \Gamma' \quad\text{iff} \quad (\exists A
  \in \rho_\text{hybrid}(\Gamma)) \; \Gamma  \oclos{A}{\REL}_{A} \Gamma'.
\end{align*}
\end{definition}

From the viewpoint of concurrency, the hybrid opinion dynamic model
has the maximum degree of non-determinism possible. Moreover, this
model is more general than the De Groot model.

\begin{theorem}\label{thm.models.degroothyb}
  ${\REL_\text{DeGroot}} \subseteq {\REL_\text{hybrid}}$.
\end{theorem}

\begin{proof}
  It follows by noting that $\Gamma_i \in \rho_\text{hybrid}(\Gamma)$
  and, for each vertex $u \in V$, the equality
  $\mu_\text{DeGroot}(\Gamma, \Gamma_i, u) = \mu_\text{hybrid}(\Gamma,
  \Gamma_i, u)$ holds.
\end{proof}

It is not necessarily the case that ${\REL_\text{gossip}} \subseteq
{\REL_\text{hybrid}}$. This is because the update functions do not
always agree when the collection of selected influences $A$ is a
singleton. In particular, for each singleton $A = \{\oobj{(v,
  u)}{i_{vu}}\}$, $\mu_\text{hybrid}(\Gamma, A , u) = o_v$, meaning
that agent $u$ in the hybrid model behaves always like a puppet when
$u \neq v$. Note that this is not (necessarily) the case in
${\REL_\text{gossip}}$.
Nevertheless, there is a class of graphs for which this inclusion
holds.

\begin{theorem}\label{thm.models.gossiphyb}
  If $G$ is such that each vertex has a self-loop and is influenced at
  most by another vertex, and the summation of its incoming influences
  is 1, then ${\REL_\text{gossip}} \subseteq {\REL_\text{hybrid}}$.
\end{theorem}
\begin{proof}
If $\Gamma \REL_\text{gossip} \Gamma'$, there is a singleton $A \in
  \rho_\text{gossip}(\Gamma)$ such that $\Gamma \AREL_A
  \Gamma'$. Let $A = \{\oobj{(v,
    u)}{i_{vu}}\}$. If $u$ has exactly one incoming edge, then $v = u$
  (by the initial assumption) and $\rho_\text{gossip}(\Gamma, A, u) =
  o_u = \rho_\text{hybrid}(\Gamma, A, u)$. Since $A \in
  \rho_\text{hybrid}(\Gamma)$, it follows that $\Gamma
  \REL_\text{hybrid} \Gamma'$. 
    If $u$ has two edges, and the self-loop is taken, the case $v=u$
    is as above. Otherwise, if $u\neq v$, the same transition is obtained in the 
    hybrid model by taking $A' \in
    \rho_\text{hybrid}(\Gamma)$ where $A' = A \cup \{\oobj{(u, u)}{1 - i_{vu}}\}$
    (an noticing that the denominator in $\mu_\text{hybrid}$ becomes $1$). 
\end{proof}
 
\section{The Framework in Rewriting Logic}
\label{sec.rl}

This section presents a rewrite theory that implements the set
relations framework in \Cref{sec.rels}. Off-the-shelf definitions are
provided to instantiate the framework with opinion dynamic models,
such as the ones introduced in \Cref{sec.models}.
This section assumes familiarity with rewriting
logic~\cite{meseguer-rltcs-1992} and
Maude~\cite{clavel-maudebook-2007} (see Section \ref{sec.prem}). 
The full Maude specification supporting the set relations framework
is available at~\cite{tool}.

A rewrite theory $\mathcal{R}$ (using Maude's notation) is defined to
represent networks of agents and their opinions. The atomic relation
(\Cref{eq.models.atomic}) is defined as a non-executable rewrite rule,
and the set relation framework is implemented using the
meta-programming facilities in Maude.
The framework is parametric on an update function ($\mu$) and a
strategy ($\rho$), as explained in \Cref{sec.models}.
The rewrite theory $\mathcal{R}$ must be extended equationally to
instantiate such parameters.

\subsection{Influences, Opinions, and State}

An agent $a$ and its opinion $o_a$, and the influence of agent $a$
over agent $b$ with weight $i_{ab}$, are specified with the help of
the following sorts and function symbols:
\begin{maude}
sorts Agent Opinion Edge .
op  <_:_>   :     Agent       Float -> Opinion [ctor] .
op  <`(_,_`):_> : Agent Agent Float -> Edge    [ctor] .
\end{maude}

\noindent The user is expected to provide appropriate constructors for
the sort \code{Agent}, e.g., by extending $\mathcal{R}$ with the
subsort relation \lstinline{subsort Nat < Agent} to use natural
numbers as identifiers for agents.

Sets of agents, opinions, and edges (sorts \code{SetAgent}, \code{SetOpinion},
and \code{SetEdge} respectively) are defined as  ``\texttt{,}''-separated sets
of elements in the usual way.
A $G$-configuration $\Gamma = \Gamma_o \cup \Gamma_i$ is represented
by a term of sort \code{Network}, defining the set of agents' opinions
($\Gamma_o$) and influences ($\Gamma_i$) with the following sorts and
function symbols:

\begin{maude}
sort Network . 
op < nodes:_ ; edges:_ > : SetOpinion SetEdge -> Network [ctor] .
\end{maude}

Analyzing opinion dynamics usually requires determining the number of
interactions between agents and the time needed to reach a given
state. A term of the form ``$N~\texttt{in step: }t~\texttt{comm:}~
nc$'' of sort \code{State} represents the state of a network $N$ at
the current time-unit $t$, when a number of
interactions/communications $nc$ have taken place:

\begin{maude}
 sort State .
 op _ in step:_ comm:_ : Network Nat Nat -> State [ctor] .
\end{maude}

\subsection{Strategies and the Atomic Relation}

The atomic relation $\REL_A$ is parametric on a non-empty subset $A
\subseteq \Gamma_i$. A strategy identifies each one of such subsets at
each time-step. A \code{SetSetEdge} is a ``\texttt{;}''-separated set
of set of edges.

\begin{maude}
sort SetSetEdge . subsort NeSetEdge < SetSetEdge . 
op mt : -> SetSetEdge [ctor] .
op _;_ : SetSetEdge SetSetEdge -> SetSetEdge [ctor assoc comm id: mt] .
\end{maude}

Some distinguished \code{SetSetEdge}s include the singleton with all
the edges in the network (De Groot model), the set containing only
singletons (Gossip model) and the set of non-empty subsets of edges
(Hybrid model).

\begin{maude}
var SE : SetEdge .   var E : Edge .
op deGroot  : SetEdge     -> SetSetEdge .
eq deGroot(SE) = SE .

op gossip  : SetEdge       -> SetSetEdge .
eq gossip(empty) = mt .
eq gossip((E, SE)) = E ; gossip(SE) .

op hybrid : SetEdge -> SetSetEdge .
eq hybrid(SE) = power-set(SE) \ empty .

op strategy :  -> SetSetEdge . --- user defined strategy 
\end{maude}

\noindent
The operator \code{strategy} must be defined by the user to identify
the subsets $A \subseteq \Gamma_i$ available in each transition. This
can be done, e.g., by adding the equation

\lstinline{eq strategy =  gossip(edges) .}

\noindent
where \code{edges} is the set of edges in the network currently being
modeled.

The atomic rewrite relation is captured by a non-executable rewrite
rule that updates the belief of a given \code{AGENT} ($u$ in
\Cref{eq.models.atomic}) when a set of \code{EDGES} ($A$) is selected
and the current state of the system is \code{STATE} ($\Gamma$):

\begin{maude}
var AGENT : Agent . vars BELIEF BELIEF' : Float . var STATE : State . 
vars SETEDGE EDGES : SetEdge .

op update : State SetEdge Agent -> Float .  --- user defined $\mu$

crl [atomic] : < AGENT : BELIEF >  =>  < AGENT :  BELIEF' >
    if BELIEF' := update(STATE, SETEDGE, AGENT) [nonexec] .
\end{maude}

The function \code{update} ($\mu$ in \Cref{eq.models.atomic}) must be
specified by the user. The framework provides instances of this
function for the models presented in \Cref{sec.models}.

An asynchronous, parallel, or synchronous rewrite step, depending on
the underlying strategy, is captured by the rewrite rule \code{step}
below:

\begin{maude}
var SETNODE : SetNode . vars STEPS COMM : Nat . 
op moduleName : -> Qid . --- Name of the module with the user's network 

crl [step] : STATE => STATE' 
    if  EDGES ; SSE := strategy /\
        STATE'      := step([moduleName], STATE, EDGES) .
\end{maude}

In this rule, the current \code{STATE} is updated to \code{STATE'} by
non-deterministically selecting a set of \code{EDGES} from the set of
set of edges available according to the \code{strategy}.  The function
\code{step} below takes as parameters the meta-representation of the
user's module defining the network (\code{moduleName}), the current
state, and the selected set of edges.

\begin{maude}
var SETAG : SetAgent . var SETOP : SetOpinion . var OP : Opinion .
op step : Module State SetAgent SetOpinion SetEdge -> State .
op step : Module State                     SetEdge -> State .
eq step(M, STATE, EDGES) = step(M, STATE, $\highlight{incidents}$(EDGES), $\highlight{empty}$, EDGES) .
eq step(M, STATE, $\highlight{empty}$, $\highlight{SETOP}$, EDGES) =  
   < nodes: ($\highlight{nodes(STATE) / SETOP}$) ; edges: edges(STATE) > 
   in step: (steps(STATE) $\highlight{+ 1}$) comm: (comm(STATE) $\highlight{+ | non-self(EDGES) |}$) .
eq step(M, STATE, $\highlight{(AGENT, SETAG)}$, SETOP, EDGES) = 
   step(M, STATE, SETAG, (SETOP, $\highlight{next(M, AGENT, EDGES, STATE)}$), EDGES) .

\end{maude} 

\noindent
The  function \code{step} recursively computes the beliefs of the
agents \code{incident} to \code{EDGES}. The updated beliefs are
accumulated in the set of opinions \code{SETOP}. The opinions of the
other agents remain as in \code{STATE} (operator \code{/}), and the
number of steps and the number of communications are updated
accordingly. The expression \code{| non-self(.) |} returns the number
of edges that are not self-loops, and 
\code{nodes(.)} returns the opinions ($\Gamma_o$) in a state.

The  function \code{next} computes the outcome of the transition
$\oobj{u}{o_u} \REL_A \oobj{u}{o_u'}$ by applying (\code{metaApply})
the rule \code{atomic} with the needed substitutions to make this rule
executable (and deterministic). Namely, it fixes the opinion to be
updated (\code{AGENT} and \code{BELIEF}), the current \code{STATE} and
the set of \code{EDGES} to be considered during the update.

\begin{maude}
op next : Module Agent SetEdge State -> Opinion . 
ceq  next(M, AGENT, EDGES, STATE) = OP 
 if  SUBS := 'AGENT:Agent   <- upTerm(AGENT) ; 
             'BELIEF:Float  <- upTerm(opinion(AGENT, STATE)) ; 
             'STATE:State   <- upTerm(STATE) ; 
             'EDGES:SetEdge <- upTerm(EDGES) /\ 
     RES? := metaApply(M, upTerm(< AGENT :  opinion(AGENT, STATE) >), 
             $\highlight{'atomic, SUBS, 0}$) /\
     OP   := if RES? == failure then error 
             else downTerm(getTerm(RES?), error) fi .
\end{maude}

The \code{opinion} function returns the opinion of an agent in a given
state.
 
\section{Experimentation}
\label{sec.exp}

This section shows how Maude and some of its tools can be used to
analyze instantiated versions of the rewrite theory $\mathcal{R}$ (see
\Cref{sec.rl}) to better understand the evolution of opinions in
networks of agents.
Of special interest is checking the (im)possibility of reaching a
consensus (i.e., agent's opinions converge to a given value) or
stability of the systems, computing the number of steps to reach
consensus, computing an optimal strategy to reach consensus, measuring
the polarization of the system at each time-step, among others.
It is noticed that for De Groot and Gossip-like models, there are
theoretical results identifying topological conditions that guarantee
consensus. In particular, in these models, the agents reach consensus
if the graph is strongly connected and aperiodic (i.e., the greatest
common divisor of the lengths of its cycles is one)~\cite{Golup2017}.

\subsection{Finding Consensus}

Let \code{Example-DG} be the module/theory extending $\mathcal{R}$
with the following operators and equations:

\begin{maude}
op init : -> Network .                 --- Initial state (as in Fig 1)
eq init = < nodes: ... ; edges: ... > in step: 0 comm: 0 .
eq moduleName = 'Example-DG .          --- Name of the theory
--- Predefined $\mu$ for De Groot
eq update(STATE, SETEDGE, AGENT)  = deGrootUpdate(STATE, SETEDGE, AGENT) .        
eq strategy   = deGroot(edges(init)) . --- De Groot strategy
\end{maude}

The following command answers the question of whether it is possible
to reach a consensus from the \code{init}ial state. (Function
\code{consensus(.)} checks if all opinions $o_i$ and $o_j$ in a given
state satisfy $|o_i - o_j| < \epsilon$, where $\epsilon$ is an error
bound).

\begin{maude}
Maude> search [1]  init =>* STATE such that consensus(STATE) .

Solution 1 (state 34)
STATE --> < nodes: < 0 : 4.80e-1 >, < 1 : 4.79e-1 >, < 2 : 4.79e-1 >, ...
            edges: <(0,1): 5.99e-1 >, <(0,2): 4.00e-1 >, ... >
            in step: 34 comm: 272 
\end{maude}

The consensus about the given proposition is approximately $0.48$ and
it is reached in 34 steps. 
Since in the De Groot model all the 12 edges
are considered in each interaction, there is a total of 272 = $34\times 8$ 
communications (the interactions on the self-loops are not considered in 
that counting). 
 Note that an application of rule \code{step} in this
case is completely deterministic (the strategy considers only one
possible outcome, including all the edges of the network).

Let \code{Example-H} be as \code{Example-DG}, but considering the
strategy and update functions for the hybrid model.  As explained in
\Cref{sec.models.hybrid}, the Hybrid model exhibits the maximum degree
of non-determinism.  Using \lstinline{search} to check the existence
of a reachable state satisfying consensus for the system in
Figure~\ref{fig.rels.vaccines} (12 edges) becomes unfeasible: a state
may have up to 4095 (non-empty subsets of $\Gamma_i$) successor
states. Certainly, for this network, a solution must exist due to
the above output of the \lstinline{search} command and the fact that
${\REL_\text{DeGroot}} \subseteq {\REL_\text{hybrid}}$.

Consider the following rewrite rule and expression in Maude's strategy language:

\begin{maude}
crl [step'] :   STATE  => STATE'
if  STATE'  := step([moduleName], STATE, $\highlight{EDGES}$) [$\highlight{nonexec}$] .

var STR : SetSetEdge .
strat round :  SetSetEdge @ State . 
sd round(EDGES ; STR) := (match STATE s.t. $\highlight{consensus}$(STATE))
                          or-else $\highlight{step'[EDGES <- EDGES]}$ ; $\highlight{round(STR)}$ .
\end{maude}

\noindent
Unlike \code{step}, rule \code{step'} does not use the model
strategy to select the set of \code{EDGES} that will be used to
compute the next state (and hence, it is non executable).
The Maude's strategy \code{round}  checks whether the current
state satisfies consensus and stops. Otherwise, it
non-deterministically chooses a set \code{EDGES}, applies the rule
\code{step'} instantiating the set of edges with that particular set,
and it is recursively called without \code{EDGES}. In other words,
\code{round} starts with a set of possible interactions and it allows
for these interactions to happen only once. This is certainly one of
the possible behaviors that can be observed with the Hybrid
model.
Using this strategy, it is possible to find some states that satisfy
consensus  and answer the question whether by selecting some
groups of agents (non necessarily disjoint) that, interacting only
once, may lead to a consensus.  (Function \code{filter>=(n,STR)}
returns the sets in \code{STR} with cardinality at least $n$). 

\begin{maude}
Maude> dsrew [1] init using round(hybrid(edges)) .
Solution 1
result State: < nodes: < 0 : 0.0 >, ... edges: ... > in step: 8 comm: 13 .

Maude> dsrew [1] init using round(filter>=(6, hybrid(edges))) .
Solution 1
result State: < nodes: < 0 : 1.50e-1 >, ... > in step: 21 comm: 88 .
\end{maude}

As expected, because of the non-deterministic nature of the Hybrid model, the
value of consensus (and the number of steps to reach such a state) can heavily
depend on the choice of edges at each step. In the first output returned by
\lstinline{dsrew} in the first command, 
all the sets considered by 
\code{round} included edges 
where $a$ acts as an influencer and the edge $f\to a$ is never selected. This
explains the value of the consensus, where the opinion of $a$ was propagated to
her neighbors. In the send command, larger groups are chosen to interact, and
the edge $f\to a$ is selected in  4 out of the 21 interactions. Hence, $a$
eventually changes her opinion.

\subsection{Statistical Analysis}

An alternative approach to deal with the inherent state explosion
problem when analyzing $\mathcal{R}$ is to perform statistical model
checking.
In the following, the tool
\code{umaudemc}~\cite{DBLP:journals/jlap/RubioMPV21} is used for such
a purpose.  The \code{umaudemc} command \code{scheck} enables
Monte-Carlo simulations of a rewrite theory extended with
probabilities; it estimates the value of a quantitative temporal
expression written in the query language
QuaTEx~\cite{agha-pmaude-2006}.

Consider the following QuaTEx expression that computes the probability
of reaching a consensus before $N$ communications:
\begin{maude}
Prob-C(N) = if (s.rval("consensus(S)")) then 1.0 else 
             if (s.rval("comm(S)") <= N) then # Prob-C(N) else 0.0 fi fi;
\end{maude}

The two  commands below estimate the probability of reaching consensus
before 30 (QuaTEx formula \code{E[Prob(30)]})
and 20 communications, respectively, in the running example  when 
the gossip-based model is considered.
 The confidence level
of these analyses is $95\%$ and the same probability is assigned to
every successor state (\code{uniform}).

\begin{maude}
umaudemc  scheck ex-gossip init formula -a 0.05 -d 0.01 --assign uniform 
  ($\mu$ = 0.587)
umaudemc scheck ex-gossip init formula -a 0.05  -d 0.01 --assign uniform 
  ($\mu$ = 0.348)
\end{maude}

As expected, reducing the maximum number of communications decreases the changes
of reaching a consensus state. 

The authors in \cite{Bramson2017} hypothesize that the less dispersed opinion becomes, the easier it will be to reach consensus. In fact the variance, a standard
measure of dispersion, is  used as a  measure of opinion polarization in social
networks \cite{Bramson2017}. The following commands aim at testing such a 
hypothesis in the running example when considering the hybrid model: 

\begin{maude}
umaudemc  scheck example-H init ... --assign uniform 
  ($\mu$ = 0.901)
umaudemc scheck example-H init  ... --assign "term(variance(L,R))"
  ($\mu$ = 1.0)
umaudemc scheck example-H init ...  --assign "term(distance(L,R))"
  ($\mu$ = 0.987)
\end{maude}

These commands estimate 
the probability of reaching consensus before 300 communications (\code{E[Prob(300)]}). 
In the first case, all the successor states are assigned the same
probability. In the second, successor states whose set of chosen 
agents has  higher \code{variance} 
are assigned higher probabilities. In the third command, 
successor states whose set of chosen agents are more \emph{polarized}, 
in the sense that the \code{distance} between the maximal and the minimal 
opinions is bigger, are assigned higher probabilities. 
These results confirm the hypothesis that
 it is more likely (1.0 vs 0.9)  to reach consensus sooner when 
communications of agents with more distant opinions is encouraged to reduce 
dispersion of opinions.

\section{Concluding Remarks}
\label{sec.concl}

This paper presented a unified framework for dynamic opinion models. Such
 models are tools to analyze the evolution
of opinion values, about a given topic, in a network of agents whose
opinion may be influenced by other agents.
Set relations, which are used for specifying and analyzing
concurrent behavior in collections of agents, are the formalism used
to unify the modeling of these systems.
This framework relies on two mechanisms,
namely, an atomic relation that updates the opinion of single
agents based on a collection of interactions and a strategy defining
the collections of interactions to be considered.
The framework is formally specified as a rewrite theory, which is
expected to be instantiated for the opinion dynamic model of interest.
Three different dynamic opinion models (De Groot, goossip-like, and
hybrid) are shown to be instances of this framework.
Experiments on these models  show that statistical model checking is
a promising alternative to tackle the state explosion problem when
analyzing models with a high degree of non-determinism, such is the
case of the hybrid model.
To the best of the authors' knowledge, this is the first documented
effort to make available concurrency theory, techniques, and tools for
the specification and analysis of opinion dynamics models and
properties such as polarization and consensus.

The ultimate goal of making available computational ideas and
approaches for analyzing phenomena in social networks requires
(significant) additional work. 
First, a more in-depth exploration of properties related to
these phenomena in social networks is required. This may lead to the
proposal of new temporal and probabilistic properties that cannot be
handled with current techniques and approaches supporting the opinion
dynamic modeling community, but that may be highly supported by the
developments in concurrency and computational logics.
Second, extensions to the current framework in terms of more general
dynamic networks (i.e., the value of influences can change), temporal
networks (i.e., nodes and edges can appear and disappear), and the
inclusion of several topics/propositions that may share causal relations are in
order.
Third, more experimental validation is required, ideally with data
gathered from real social networks.
Fourth, building on the abstract relations proposed here, 
techniques from concurrency theory become available for 
the analysis of social systems. It is worth exploring standard concurrency techniques such as bisimulation and testing equivalences to answer questions such as whether two social systems ought to be equivalent and whether there is a social context, represented as a social system, that can tell the  difference between two other social systems.

\bibliographystyle{splncs04}

\end{document}